\documentclass{ifacconf}

\usepackage{graphicx}      
\usepackage{natbib}        

\usepackage{graphicx}
  \graphicspath{{./Pics/}}
  \DeclareGraphicsExtensions{.pdf}
\usepackage{subfigure}
\usepackage{epstopdf}
\usepackage{float}
\usepackage{amsmath}
\usepackage{amsfonts}
\usepackage{epsfig}
\usepackage{graphicx}
\usepackage{arydshln}
\usepackage{algorithm}
\usepackage{algorithmicx}
\usepackage{algpseudocode}
\usepackage{verbatim}
\usepackage{subfigure}
\usepackage{enumerate}
\usepackage{rotating}
\usepackage{threeparttable}
\usepackage{caption}
\usepackage{booktabs} 
\usepackage{IEEEtrantools}
\usepackage{amssymb,latexsym}
\usepackage{mathrsfs}
\usepackage{dsfont}
\usepackage{caption}
\usepackage{color}
\usepackage{multirow}

\newenvironment{proof}{
  \noindent\textbf{Proof.} 
}{
  \hfill\rule{1.5mm}{1.5mm}\par
}
\def \T{{\mbox{\tiny T}}}
\def \argmin{\operatorname*{argmin}}

\usepackage{tikz}
\usepackage[many]{tcolorbox}
\usetikzlibrary{calc}

\usetikzlibrary{positioning,shapes}
\usetikzlibrary{arrows}

\tcbuselibrary{skins}
\newtcolorbox{resp}[1][]{%
	enhanced jigsaw,%
	colback=gray!5!white,%
	colframe=gray!80!black,%
	size=small,%
	boxrule=1pt,%
	halign title=flush center,%
	coltitle=black,%
	breakable,%
	drop shadow=black!50!white,%
	attach boxed title to top left={xshift=1cm,yshift=-\tcboxedtitleheight/2,yshifttext=-\tcboxedtitleheight/2},%
	minipage boxed title=3cm,%
	boxed title style={%
		colback=white,%
		size=fbox,%
		boxrule=1pt,%
		boxsep=2pt,%
		underlay={%
			\coordinate (dotA) at ($(interior.west) + (-0.5pt,0)$);
			\coordinate (dotB) at ($(interior.east) + (0.5pt,0)$);
			\begin{scope}[gray!80!black]
				\fill (dotA) circle (2pt);
				\fill (dotB) circle (2pt);
			\end{scope}
		}%
	},%
	#1%
}

\newtheorem{theorem}{Theorem}[section]

\newtheorem{corollary}[theorem]{Corollary}

\newtheorem{definition}[theorem]{Definition}
\newtheorem{example}{Example}

\newtheorem{remark}[theorem]{Remark}
\newtheorem{assumption}[theorem]{Assumption}
\numberwithin{equation}{section}

\begin{document}
\begin{frontmatter}
\title{Hierarchical Control for Continuous-time Systems via General Approximate Alternating Simulation Relations\thanksref{footnoteinfo}} 

\thanks[footnoteinfo]{
This work is supported by the Guangdong Provincial Project (No. 2024QN11X053), the Guangdong Basic and Applied Basic Research Foundation (No. 2026A1515010222) and by the Youth S$\&$T Talent Support Programme of GDSTA (No. SKXRC2025468). (Corresponding author: Bingzhuo Zhong.)
}

\author[HKUST(GZ)]{Zhiyuan Huang},
\author[HKUST(GZ)]{Shuo Li},
\author[UCB]{Murat Arcak},
 \author[UC]{Majid Zamani},
 \author[HKUST(GZ)]{Bingzhuo Zhong}


\address[HKUST(GZ)]{Thrust of Artificial Intelligence, The Hong Kong University of Science and Technology (Guangzhou), China (e-mail: zhuang655@connect.hkust-gz.edu.cn; bingzhuoz@hkust-gz.edu.cn; sli430@connect.hkust-gz.edu.cn ).}
\address[UCB]{Department of Electrical Engineering and Computer Sciences, University of California, Berkeley, USA (e-mail: arcak@berkeley.edu)}
\address[UC]{Department of Computer Science, University of Colorado Boulder, USA (e-mail: majid.zamani@colorado.edu)}

\begin{abstract}                
This paper introduces a general approximate alternating simulation relation (\emph{$\varepsilon$-gAAS relation})  for continuous-time systems, which relaxes existing simulation relations to tolerate larger mismatches between abstract and concrete models. 
The definition of gAAS for continuous-time systems is first proposed, and its properties are investigated.
Then, a control refinement method is developed to enable hierarchical control for the gAAS relation. Finally, case studies demonstrate the effectiveness of the proposed approach, highlighting its advantages over existing methods.

\end{abstract}


\end{frontmatter}

\section{Introduction}
\vspace{-5pt}
In recent decades, cyber-physical systems (CPS) have emerged as a key technology enabling the development of a smarter, more interconnected world. As these systems and their task specifications grow ever more complex, designing suitable controllers has become increasingly difficult, which in turn has spurred interest in hierarchical control as a way to tackle this problem.

The general structure of hierarchical control is depicted in Fig. \ref{fig:Hierarchical}.
The system that is actually being controlled is called the \emph{concrete system}, denoted by $\Sigma$, whereas its simplified model is referred to as the \emph{abstract system}, denoted by $\hat{\Sigma}$.
Instead of designing a controller directly for the concrete system, hierarchical control first develops a high-level controller for the abstract system and then translates it to the concrete system via control refinement. This procedure guarantees that the desired behavior of the concrete system is obtained by maintaining specified relations between the abstract and concrete systems.
While hierarchical control streamlines controller design by concentrating on the abstract model of the system, devising robust control refinement techniques that preserve the necessary relations continues to be a major challenge.
\begin{figure}[htbp]
	\centering
	\includegraphics[width=0.38\textwidth]{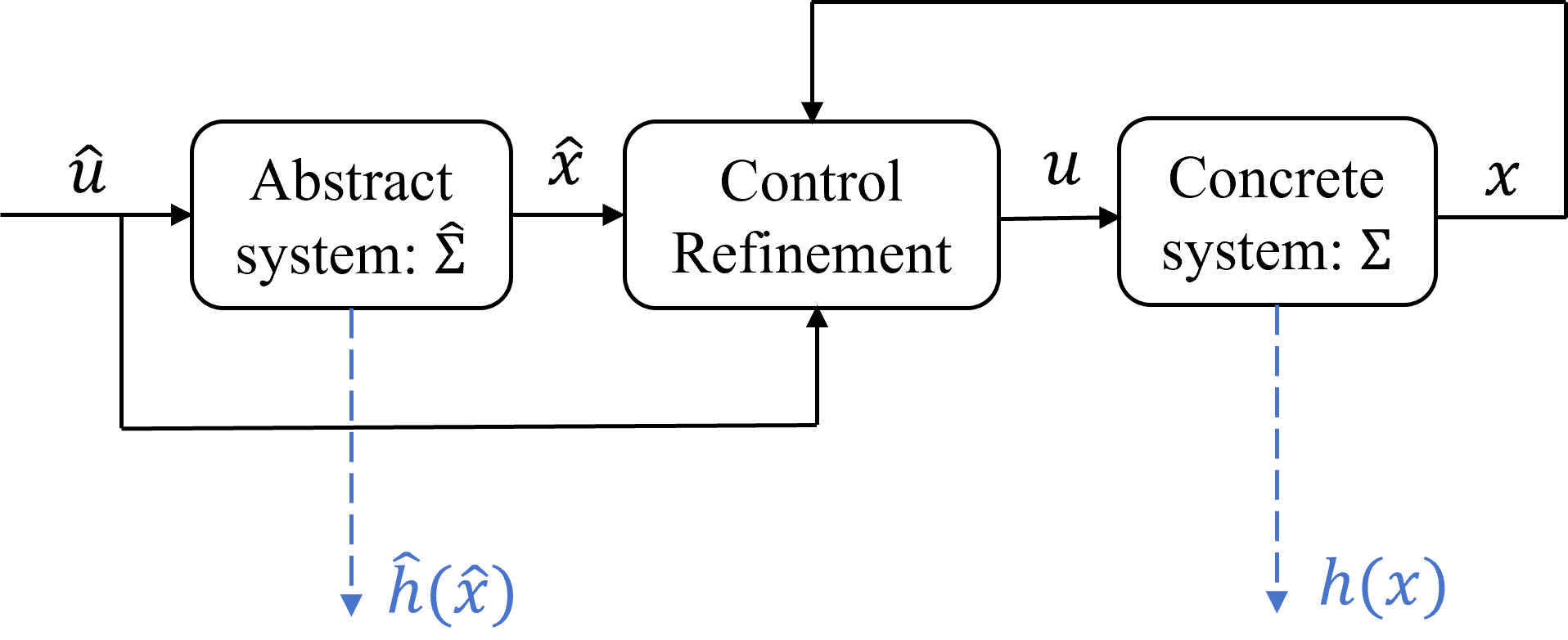}
    \vspace{-4pt}
	\caption{Hierarchical control system architecture.}
    \vspace{4pt}
	\label{fig:Hierarchical}
\end{figure}

In recent years, different types of system relations and their associated control refinement methods have been explored.
The framework of exact bisimulation relations was utilized in \citep{tabuada2009verification} and has since been investigated for discrete event systems in \citep{farhat2020control, zhang2017infinite} and for continuous systems in \citep{pappas2003bisimilar, van2004equivalence}.
Furthermore, the notion of the feedback refinement relation was developed in \citep{reissig2016feedback}, where a symbolic control refinement approach was also proposed.

Beyond the previously discussed relations, a more relaxed notion, termed the approximate (alternating) simulation relation, was proposed in \citep{girard2007approximation, girard2006hierarchical}.
The key idea behind the approximate simulation relation is that, as long as the outputs of the concrete and abstract systems remain within a prescribed error bound, this closeness is preserved over time or along system transitions through suitable control refinement.
For continuous-time control systems, a standard control refinement strategy under an approximate simulation relation is to construct a Lyapunov-like simulation function together with an interface function that maps inputs of the abstract system to those of the concrete system, thereby guaranteeing the desired stability properties for the interconnected concrete–abstract system.
In line with this concept, several simulation functions and their corresponding interface functions based on approximate simulation relations have been introduced for both linear systems \citep{girard2009hierarchical,girard2011approximate} and nonlinear systems \citep{smith2019continuous}.
Nevertheless, the existing conditions for using Lyapunov-like simulation functions in control refinement under approximate (alternating) simulation relations remain conservative, which in turn imposes limitations on achievable control performance \citep{zhong2024hierarchical}.

In recent years, several variants of the approximate simulation relation have been put forward, including approximate probabilistic relations in \citep{zhong2023automata} and neural simulation relations in \citep{nadali2024transfer}.
Beyond these developments, a more relaxed notion, called the general approximate alternating simulation relation, was presented in \citep{zhong2024hierarchical}, where control refinement techniques were also explored for both discrete-event systems and discrete linear systems.
Nonetheless, the class of system models examined in \citep{zhong2024hierarchical} remains rather restricted.

Motivated by the limitations of existing studies, this paper focuses on constructing simulation functions and addressing the associated control refinement problem under a general approximate alternating simulation relation for continuous-time systems.  
The main contributions are summarized as follows:  
(i) the general approximate alternating simulation relation (gAAS) for continuous-time systems is formally introduced and its properties are analyzed;  
(ii) Lyapunov-like simulation functions tailored to this general approximate alternating simulation relation are proposed for continuous-time systems, together with the corresponding control refinement methods.

The main contributions of this work, relative to existing literature, can be outlined as follows.
In contrast to \citep{zhong2024hierarchical}, this paper examines the construction of simulation functions and the associated control refinement under a general approximate simulation relation for continuous-time systems subject to input constraints.
Compared with \citep{girard2009hierarchical}, the proposed control refinement framework allows more relaxed input constraints on the abstract system and additionally incorporates the possibility of discontinuous inputs in the abstract model. 

The remainder of the paper is structured as follows.
Section \ref{Sec:Preliminary} presents the preliminaries, including the notation, models, and problem formulation, together with a motivating example.
Section \ref{section: gaas} introduces the general approximate alternating simulation relation and discusses its main properties.
Section \ref{Hierarchical_linear} proposes a hierarchical control framework built upon the simulation function.
Section \ref{case study} reports case studies that illustrate the effectiveness of the proposed method.
Finally, Section \ref{conclusion} summarizes and concludes the paper.

\section{Preliminary and problem formulation} \label{Sec:Preliminary}

\subsection{Notations}
\vspace{-6pt}
We denote the set of real numbers by $\mathbb{R}$.
Subscripts are employed to specify restrictions.
For example, $\mathbb{R}_{>0}$ denotes the set of positive real numbers.
For $a, b \in \mathbb{R}$ with $a \leq b$, the closed, open, and half-open intervals in $\mathbb{R}$ are denoted by $[a,b]$, $(a,b)$, $[a,b)$, and $(a,b]$, respectively.
The notation $\mathbb{R}^{n}$ refers to the $n$-dimensional Euclidean vector space. 
The Cartesian product of sets $X_i$, $i \in \{1, \dots, N\}$, is defined as $X_1 \times \ldots \times  X_N =\{(x_1,\cdots,x_N) \mid x_i \!\in\! X_i, i\in \{1, \dots, N\} \}$. 
Given sets $X_i$, $i \in \{1, \dots, N\}$, and their Cartesian product $X_1 \times \cdots \times X_N$, 
the projection onto $X_i$ is denoted by the mapping $\text{Proj}_{X_i}: X\rightarrow X_i$.
For a matrix $A$, we use $A^{\T}$ to denote its transpose, and for positive semidefinite $A$, $A^{1/2}$ denotes its matrix square root.
Denote by $\| \cdot\|$ the Euclidean norm for vectors and the spectral norm for matrices.
Let $\mathcal{B}_{m}(b)$ be the closed Euclidean ball in $\mathbb{R}^m$ centered at the origin with radius $b$, that is, $\mathcal{B}_{m}(b) := \{ x \in \mathbb{R}^m : \|x\| \le b \}$.
Finally, the notation $A \succ 0$ (resp. $A \succeq 0$) indicates that $A$ is positive definite (resp. positive semidefinite).

\subsection{Preliminary}
In this paper, we are interested in continuous-time control systems, as described in the following definition.
\begin{definition}\label{def:ct-system}
A \emph{continuous-time control system} is defined as a tuple $\Sigma := (X, X_0, U, \mathsf{U}, f, Y, h)$, where $X\subseteq\mathbb{R}^{n}$ is the continuous state space, $X_0 \subseteq X$ is the set of initial states, $U$ is the continuous input space, and $Y$ is the output space. 
Additionally, 
the system dynamics are modeled as $f:  X\times U\to 2^{\mathbb{R}^{n}}$, and the output function is given by $h: X \to Y $.
The admissible input map is 
\(\mathsf{U}: X \to U\), where \(\mathsf{U}(x)\) specifies the set of admissible inputs at state \(x\). 
In particular, if the input constraints are state-independent, one can write \(\mathsf{U}(x)\equiv U\) for all \(x\in X\).
\end{definition}

According to the Definition \ref{def:ct-system}, the continuous-time control system $\Sigma$ can be equivalently described by: 
\begin{align}\label{eq:Ct-system-general}
    \Sigma:\left\{
    \begin{array}{rl}
	\dot{x}(t) \in & f(x(t),u(t)),\\
	y(t)=&h(x(t)),
    \end{array}
	\right.
\end{align}
with $x(t)\in X$, $u(t)\in \mathsf{U}(x(t))\subseteq U$, and $y(t)\in Y$.
A trajectory of the system corresponding to a measurable input $u: [t_0,+\infty) \to U$ is a measurable function defined as $\xi_{x,u}: [t_{0},+\infty) \to X$ with initial conditions $\xi_{x,u}(t_{0}) \in X_{0}$ and $t_{0}\in \mathbb{R}_{\ge0}$, satisfying $\dot{\xi}_{x,u}(t) \in f(\xi_{x,u}(t), u(t))$ for almost everywhere (a.e.) $t \in [t_0,+\infty)$.
Since $f$ is set-valued, there may exist multiple trajectories corresponding to the same initial condition and input. 
The set of all such trajectories is denoted by
\begin{align}
   & \mathcal{T}_{X_0}(u) = \{ \xi_{x,u}: [t_0,+\infty) \to X \mid \dot{\xi}_{x, u} \in f(x, u) \nonumber \\
   & ~~~~~~~~~~~~~~~~~~~~~\text{ a.e. } t \in [t_0,+\infty),~  \xi_{x,u}(t_0) \in X_0 \}.
\end{align}
Since the input $u(t)$ can have discontinuities—for example, in piecewise-constant or switching control—we define the concepts of jump time and jump value to characterize these discontinuities.
\begin{definition} \label{def: jump value and set}
Consider $u : [t_{0},+\infty) \to \mathbb{R}^{m}$ be an input that may be discontinuous. 
A time $\tau \ge t_{0}$ is called a \emph{jump time} of $u$ if the left and right limits exist and satisfy $u(\tau^{+}) \ne u(\tau^{-})$.
The corresponding \emph{jump value} of $u$ at $\tau$ is defined as $\delta(\tau) := u(\tau^{+}) - u(\tau^{-})$.
\end{definition}

Next, following \citep[Definition 4.19]{tabuada2009verification}, we introduce the notion of an $\varepsilon$-approximate alternating simulation relation between two continuous-time systems as follows.
\begin{definition}\label{definition: as}
    Consider two continuous-time systems $\Sigma = (X, X_0, U, \mathsf{U}, f,  Y, h)$ and $\hat{\Sigma}=(\hat{X}, \hat{X}_0, \hat{U},  \hat{\mathsf{U}}, \hat{f}, Y, \hat{h})$, and a constant $\varepsilon \in \mathbb{R}_{\geq 0}$.
    A relation $\mathcal{R}\subseteq \hat{X} \times X$ is an \emph{$\varepsilon$-approximate alternating simulation relation} ($\varepsilon$-AAS relation) from $\hat{\Sigma}$ to $\Sigma$, if 
    \begin{enumerate}[(i)]
        \item $\forall \hat{x}_0 \in \hat{X}_0$, $\exists x_0 \in X_0$ such that $(\hat{x}_0,x_0)\in\mathcal{R}$;
        \item $\forall t \ge t_{0}$, $\forall (\hat{\xi}_{\hat{x},\hat{u}}(t), \xi_{x,u}(t))\in\mathcal{R}$, it follows that \linebreak
        $\mathbf{d}(\hat{h}(\hat{\xi}_{\hat{x},\hat{u}}(t)),h(\xi_{x,u}(t)))\le \varepsilon$;
        \item
        $\forall (\hat{x}_{0},x_{0})\in\mathcal{R}$, $\forall \hat{u}(\cdot):[t_0,+\infty)\to \hat{U}$, $\exists u(\cdot):[t_0,+\infty)\to U$ such that $\forall \xi_{x,u} \in \mathcal{T}_{X_0}(u)$, $\exists \hat{\xi}_{\hat{x},\hat{u}} \in \mathcal{T}_{\hat{X}_0}(\hat{u})$, such that $\forall t \ge t_{0}$, $(\hat{\xi}_{\hat{x},\hat{u}}(t), \xi_{x,u}(t)) \in \mathcal{R}$,

    \end{enumerate}
    where $x_{0}: = x(t_{0})$,  $\hat{x}_{0}: = \hat{x}(t_{0})$ with $t_{0}\in \mathbb{R}_{\ge 0}$ being the initial time and $\mathbf{d}:  Y \times Y \!\rightarrow \!\mathbb{R}_{\geq 0}$ is a metric on $ Y$. 
    If there exists an $\varepsilon$-approximate alternating simulation relation from $\hat{\Sigma}$ to $\Sigma$, denoted by $\hat{\Sigma}\preceq^{\varepsilon}_{\mathcal{AS}}\Sigma $, we say that $\hat{\Sigma}$ is \emph{$\varepsilon$-approximately alternatively simulated by} $\Sigma$.
\end{definition}

\subsection{Motivation and Problem formulation}
In this subsection, we illustrate the limitation of the $\varepsilon$-AAS relation introduced in Definition \ref{definition: as} by means of a simple example.
\begin{example}
Consider the continuous-time concrete system $\Sigma_{} := (X_{}, X_{0}, U_{}, \mathsf{U}_{}, f, Y_{}, h)$, which is described by the model
\begin{equation} \label{case: linear concrete model}
  \Sigma_{}:  \left\{\begin{matrix}
\dot{x}_{}(t)
= \begin{bmatrix}
0  &1 \\
 0 &0
\end{bmatrix}x_{}(t) + \begin{bmatrix}
0   \\
 1 
\end{bmatrix}u_{}(t),\\
y_{}(t)=\begin{bmatrix}
1  &0
\end{bmatrix}x_{}(t),
\end{matrix}\right.
\end{equation}
where $x_{}(t) = [p_{}(t); v_{}(t)] \in X_{}$ is the state vector, with $p_{}(t)$ and $v_{}(t)$ denoting the position and velocity, respectively, and $u_{}(t) \in U $ denoting the acceleration serving as the control input.
Consider the corresponding abstract system $\hat{\Sigma}_{} = (\hat{X}_{}, \hat{X}_{0}, \hat{U}_{}, \hat{\mathsf{U}}_{}, \hat{f}, Y_{}, \hat{h})$ defined as 
\begin{equation}\label{case: linear abstract model}
  \hat{\Sigma}_{}:  \left\{\begin{matrix}
 \dot{\hat{x}}_{}(t) = \hat{u}_{}(t),\\
\hat{y}_{}(t)=\hat{x}_{},
\end{matrix}\right.
\end{equation}
where $\hat{x}_{}(t) = \hat{p}_{}(t)$ denotes the position, and $\hat{u}_{}(t)\in  \hat{U}$ denotes the velocity, which serves as the control input of $\hat{\Sigma}_{}$.
To characterize the similarity between $\Sigma$ and $\hat{\Sigma}_{}$ with an approximate alternating simulation relation as in Definition \ref{definition: as}, we consider the output error (output metric) $e(t)= \mathbf{d}(y,\hat y) =y(t) -\hat{y}(t) $.
From \eqref{case: linear concrete model} and \eqref{case: linear abstract model}, the error dynamics can be expressed as
\begin{equation*}
\dot{e}(t) = v(t) - \hat{u}(t), \quad \dot{v}(t) =u (t).
\end{equation*}
One can see that fast or discontinuous variations in the input $\hat{u}(t)$ directly influence the error rate $\dot{e}(t)$, which in turn forces the velocity $v(t)$ to adjust in order to keep the error $e(t)$ small.
However, because $v(t)$ can change only through the bounded acceleration $u(t) \in U$, it cannot react quickly enough to follow such abrupt variations in $\hat{u}(t)$.
As a consequence, the discrepancy $\dot{e}(t) = v(t) - \hat{u}(t)$ becomes large whenever $\hat{u}(t)$ changes rapidly.
Thus, the output mismatch cannot be kept uniformly below a small constant.
Equivalently, the parameter $\varepsilon$ in condition \textit{(ii)} of Definition \ref{definition: as} must take a large value when a rapidly varying input $\hat{u}$ is used, which complicates the control refinement.

\end{example}

Motivated by the constraints of the $\varepsilon$-AAS relation, this paper introduces a more relaxed form of alternating simulation relation for continuous-time systems, referred to as the general approximate alternating simulation relation, and examines how it relates to the conventional approximate alternating simulation relation. We then construct a hierarchical control framework for linear continuous-time systems based on the proposed general approximate alternating simulation relation, making use of the newly introduced simulation functions together with the interface function.

\section{General Approximate Alternating Simulation Relation} \label{section: gaas}
\vspace{-5pt}
The proposed general approximate alternating simulation relation for continuous-time systems is defined as follows.
\begin{definition}\label{def:cs-gAAS}
    Consider two continuous-time systems $\Sigma = (X, X_0, U, \mathsf{U}, f, Y, h)$ and $\hat{\Sigma}=(\hat{X}, \hat{X}_0, \hat{U},\hat{\mathsf{U}}, \hat{f},Y,\hat{h})$, and a constant $\varepsilon \in \mathbb{R}_{\geq 0}$.
   A relation $\mathcal{R}_{g} \subseteq X \times \hat{X} \times \hat{U}$ is said to be a \emph{general $\varepsilon$-approximate alternating simulation relation} (\emph{$\varepsilon$-gAAS relation}) from $\hat{\Sigma}$ to $\Sigma$ if the following conditions hold:
    \begin{enumerate}[(i)]
        \item $\forall \hat{x}_0 \in \hat{X}_0$, $\exists x_0 \in X_0$, $\exists \hat{u}_{0} \in \hat{U}$ such that $(x_0,\hat{x}_0, \hat{u}_{0})\in\mathcal{R}_{g}$;
        \item $\forall t \ge t_{0}$, $\forall ( \xi_{x,u}(t), \hat{\xi}_{\hat{x},\hat{u}}(t), \hat{u}(t))\in\mathcal{R}_{g}$, it follows that $\mathbf{d}(h(\xi_{x,u}(t)), \hat{h}(\hat{\xi}_{\hat{x},\hat{u}}(t)))\leq \varepsilon$;
        \item
         $\forall (x_{0},\hat{x}_{0}, \hat{u}_{0})\in\mathcal{R}_{g}$, $\exists \hat{u}(\cdot):[t_0,+\infty)\to \hat{U}$, $\exists u(\cdot):[t_0,+\infty)\to U$ such that $\forall \xi_{x,u}\in \mathcal{T}_{X_0}(u)$, $\exists \hat{\xi}_{\hat{x},\hat{u}}\in \mathcal{T}_{\hat{X}_0}(\hat{u})$, such that $\forall t \ge t_{0}$, $ ( \xi_{x,u}(t), \hat{\xi}_{\hat{x},\hat{u}}(t), \hat{u}(t)) \in \mathcal{R}_{g}$,
    \end{enumerate}
     where $x_{0}: = x(t_{0})$,  $\hat{x}_{0}: = \hat{x}(t_{0})$, $\hat{u}_{0}: = \hat{u}(t_{0})$ with $t_{0}\in \mathbb{R}_{\ge 0}$ being the initial time and  $\mathbf{d}:  Y \times Y \!\rightarrow \!\mathbb{R}_{\geq 0}$ is a metric on $ Y$. 
    If there exists an $\varepsilon$-gAAS relation from $\hat{\Sigma}$ to $\Sigma$, denoted by $\hat{\Sigma}\preceq^{\varepsilon}_{g\mathcal{AS}}\Sigma$, $\hat{\Sigma}$ is said to be \emph{generally $\varepsilon$-approximately alternatively simulated by $\Sigma$}.
\end{definition}
Definition~\ref{def:cs-gAAS} implicitly guarantees that for any $(x_0,\hat{x}_0,\hat{u}_0)\in\mathcal{R}_{g}$, there exists an interface function $u : X \times \hat{X} \times \hat{U} \to U$ such that $\hat{\Sigma}\preceq^{\varepsilon}_{g\mathcal{AS}}\Sigma$. 
We now introduce a simulation function corresponding to the $\varepsilon$-gAAS relation.
\begin{definition} \label{def: simulation function}
   Consider two continuous-time systems $\Sigma = (X, X_0, U, \mathsf{U}, f, Y, h)$ and $\hat{\Sigma} = (\hat{X}, \hat{X}_0, \hat{U}, \hat{\mathsf{U}}, \hat{f}, Y, \hat{h})$. Let $V: X \times \hat{X} \times \hat{U} \to \mathbb{R}_{\ge 0}$ be a continuously differentiable map, and define
\[
\mathcal{V} := \{ (x, \hat{x}, \hat{u}) \in X \times \hat{X} \times \hat{U} \mid V(x, \hat{x}, \hat{u}) \le \varepsilon \}.
\]
We call \(V\) a \emph{simulation function} from \(\hat{\Sigma}\) to \(\Sigma\) if the set \(\mathcal{V}\) is an \(\varepsilon\)-gAAS relation.
\end{definition}
Building on \citep[Theorem 3.2]{zhong2024hierarchical}, we present the following results to show that the $\varepsilon$-gAAS relation is strictly more expressive than the $\varepsilon$-AAS relation introduced in Definition~\ref{definition: as}.
\begin{theorem}\label{thm: g-ass more general}
    Consider continuous-time systems $\Sigma=(X, X_0, U,\mathsf{U},$ $f,Y,h)$ and $\hat{\Sigma}=(\hat{X}, \hat{X}_0, \hat{U},\hat{\mathsf{U}}, \hat{f},Y,\hat{h})$, 
    and a constant $\varepsilon \in \mathbb{R}_{\geq 0}$. 
    One has $\hat{\Sigma}\preceq^{\varepsilon}_{\mathcal{AS}}\Sigma \Rightarrow \hat{\Sigma}\preceq^{\varepsilon}_{g\mathcal{AS}}\Sigma$.
\end{theorem}
Theorem \ref{thm: g-ass more general} asserts that if there exists an $\varepsilon$-approximate alternating simulation relation from $\hat{\Sigma}$ to $\Sigma$, then an $\varepsilon$-gAAS relation from $\hat{\Sigma}$ to $\Sigma$ also exists.
However, if a system $\hat{\Sigma}$ is only known to be generally $\varepsilon$-approximately alternating simulated by $\Sigma$, this does not necessarily ensure the existence of an $\varepsilon$-approximate alternating simulation relation from $\hat{\Sigma}$ to $\Sigma$.
A counterexample is given in Section~\ref{case study}, where $\hat{\Sigma}_{} \preceq_{g\mathcal{AS}}^{0.5} \Sigma_{}$ holds, but $\hat{\Sigma}_{} \preceq_{\mathcal{AS}}^{0.5} \Sigma_{}$ fails to hold.

\vspace{-5pt}
\section{Hierarchical control for continuous-time linear systems under $\varepsilon$-gAAS relation} \label{Hierarchical_linear}
\vspace{-5pt}
In this section, we introduce a hierarchical control framework for continuous-time linear systems that satisfies the $\varepsilon$-gAAS relation.
The concrete system $\Sigma = (X, X_{0}, U, \mathsf{U}, f, $ $Y, h)$ is defined as
\begin{equation} \label{eq: linear-concrete}
  \Sigma:  \left\{\begin{matrix}
 \dot{x}(t) = Ax(t) + Bu(t)\\
y(t)=Cx(t),
\end{matrix}\right.
\end{equation}
where $x(t) \in X\subseteq \mathbb{R}^{n}$, $A \in \mathbb{R}^{n \times n}$, $u(t) \in \text{U}(x(t))=U \subseteq \mathbb{R}^{m}$, $B \in \mathbb{R}^{n\times m}$, $C\in \mathbb{R}^{p\times n}$, and $y(t) \in Y \subseteq \mathbb{R}^{p}$.
We denote the associated continuous-time abstraction of the system $\hat{\Sigma} = (\hat{X}, \hat{X}_{0}, \hat{U}, \hat{\mathsf{U}}, \hat{f}, Y, \hat{h})$ by
\begin{equation}\label{eq: linear-abstract}
  \hat{\Sigma}:  \left\{\begin{matrix}
 \dot{\hat{x}}(t) = \hat{A}\hat{x}(t) + \hat{B}\hat{u}(t)\\
\hat{y}(t)=\hat{C}\hat{x}(t),
\end{matrix}\right.
\end{equation}
where $\hat{x}(t) \in \hat{X} \subseteq\mathbb{R}^{n_{r}}$, $\hat{A}\in \mathbb{R}^{n_{r} \times n_{r}}$, $\hat{u}(t) \in \hat{\text{U}}(\hat{x}(t)) =\hat{U} \subseteq \mathbb{R}^{m_{r}}$, $\hat{B} \in \mathbb{R}^{n_{r}\times m_{r}}$, $ \hat{C}\in \mathbb{R}^{p\times n_{r}}$ and $y(t) \in Y \subseteq \mathbb{R}^{p}$, with $n_{r} \le n$, $m_{r} \le m$.

Assume that $h(x) = y(t)$ and $\hat{h}(\hat{x}) = \hat{y}(t)$, and consider $\mathbf{d}(h(x), \hat{h}(\hat{x})) = \| h(x) - \hat{h}(\hat{x})\|$.
We then introduce a method that uses a simulation function together with an explicit interface function to perform control refinement.

Consider a function of the form
\begin{equation} \label{eq: linear simulation function}
    V_{g}(x,\hat{x},\hat{u}) = \sqrt{(x - P\hat{x} - S\hat{u})^{\T} M (x - P\hat{x} - S\hat{u})},
\end{equation}
where $P \in \mathbb{R}^{n \times n_{r}}$, $S \in \mathbb{R}^{n \times m_{r}}$, and $M \in \mathbb{R}^{n \times n}$.
Define 
\begin{equation}
    \mathcal{V}_{g} : =  \{ (x,\hat{x},\hat{u}) \in X \times \hat{X}\times \hat{U} \mid V_{g}(x,\hat{x},\hat{u}) \le \varepsilon \}.
\end{equation}


Next, we define the control refinement function for the concrete system $\Sigma$, denoted by $u_{g}: X \times \hat{X} \times \hat{U}\to U$, as follows:
\begin{equation} \label{eq: linear explicit interface function}
    u_{g}(x,\hat{x},\hat{u}) =  K(x-P\hat{x}-S\hat{u})+Q\hat{x} + R\hat{u},
\end{equation}
with $K \in \mathbb{R}^{m \times n}$, $R \in \mathbb{R}^{m \times m_{r}}$, and $Q  \in \mathbb{R}^{m \times n_{r}}$.

Before presenting the sufficient conditions that ensure an $\varepsilon$-gAAS relation from $\hat{\Sigma}$ to $\Sigma$, we first introduce the following assumption. It captures the assumptions under which $V_{g}(x,\hat{x},\hat{u})$ in \eqref{eq: linear simulation function} and $u_{g}(x,\hat{x},\hat{u})$ in \eqref{eq: linear explicit interface function} qualify as a simulation function and its corresponding interface function, respectively, in the sense of Definition \ref{def: simulation function}.
\begin{assumption} \label{def: linear conditions}
    Consider the concrete continuous-time linear system $\Sigma = (X, X_{0}, U, \mathsf{U}, f, Y,h)$ in \eqref{eq: linear-concrete} and its corresponding abstract system $\hat{\Sigma} = (\hat{X}, \hat{X}_{0}, \hat{U}, \hat{\mathsf{U}}, \hat{f}, Y, \hat{h})$ in \eqref{eq: linear-abstract}, together with the function $V_{g}(x, \hat{x},\hat{u})$ and the function $u_{g}(x, \hat{x},\hat{u})$ specified as in \eqref{eq: linear simulation function} and \eqref{eq: linear explicit interface function}.
Let $\lambda_{min}(M)$ denote the smallest eigenvalue of a matrix $M$.
For some given $a_{1},b \in \mathbb{R}_{>0}$, we presume the following conditions hold:
    {
\allowdisplaybreaks
    \begin{align}
      & CP = \hat{C}, 
       ~CS =0_{p \times m_{r}},
       ~M =M^{\T}\succ 0,~
        C^{\T}C \preceq M \\
       & (A + BK)^{\T}M +  M(A+ BK)\prec -a_{1}M \label{eq: linear Lyapunov ineq}\\
       &Q^{}  =\argmin_{Q\in \mathbb{R}^{m \times n_{r}}}  \|M^{1/2} \cdot (AP -  P\hat{A} + BQ)\|  \label{eq: linear QRS select 1}  \\
       &R^{},S^{} \in \argmin_{\substack{R \in \mathbb{R}^{m \times m_r} \\ S \in \mathbb{R}^{n \times m_r}}}   \|M^{1/2}\cdot (AS + BR-P\hat{B})\|  \label{eq: linear QRS select 2}  \\
       & \|K\| \frac{\varepsilon}{\sqrt{\lambda_{min}(M)}} +\|Q \hat{x}\| + \|R\hat{u}\| \le b, ~ \mathcal{B}_{m}(b) \subseteq U  \label{eq: input constraint}
    \end{align}
    where the matrices $A$, $B$, $C$, $\hat{A}$, $\hat{B}$, and $\hat{C}$ are specified in \eqref{eq: linear-concrete} and \eqref{eq: linear-abstract}, and the matrices $P$, $S$, $M$, $K$, $Q$, and $R$ are introduced in \eqref{eq: linear simulation function} and \eqref{eq: linear explicit interface function}.
    }
\end{assumption}
Next, we derive an $\varepsilon$-gAAS relation from $\hat{\Sigma}$ to $\Sigma$ by invoking the following result, under the assumption that the abstract input $\hat{u}$ is continuous.
\begin{theorem} \label{theorem: linear valid simulation and interface}
Consider continuous-time linear systems $\Sigma$ and $\hat{\Sigma}$ described by \eqref{eq: linear-concrete} and \eqref{eq: linear-abstract}, respectively. 
Assume that the functions $V_g(x, \hat{x}, \hat{u})$ and $u_{g}(x, \hat{x}, \hat{u})$ given in \eqref{eq: linear simulation function} and \eqref{eq: linear explicit interface function} are employed, where the matrices and scalars $P$, $S$, $M$, $K$, $R$, $Q$, $a_{1}$, $b$, and $\lambda_{min}(M)$ satisfy all conditions stated in Assumption \ref{def: linear conditions}. 
Assume that
\begin{enumerate}[(i)]
    \item there exist $\bar{r}_{\max} \in \mathbb{R}_{\ge 0}$ and $a_1$ in \eqref{eq: linear Lyapunov ineq} such that $\bar{r}(\hat{x}, \hat{u}, \dot{\hat{u}}) \le \bar{r}_{\max}$ and $\frac{2\bar{r}_{\max}}{a_1} \le \varepsilon$, where
    \begin{align}
    &\bar{r}(\hat{x},\hat{u},\dot{\hat{u}}) =  \bar{r}_{1}\|\hat{x}(t)\| + \bar{r}_{2}\|\hat{u}(t)\| + \bar{r}_{3}\| \dot{\hat{u}}\|  \text{with} \nonumber \\
    &\bar{r}_{1}  =\min_{Q\in \mathbb{R}^{m \times n_{r}}}  \|M^{1/2} \cdot (AP -  P\hat{A} + BQ)\|  \label{eq: linear bounded var1} \\
    &\bar{r}_{2} = \min_{\substack{R \in \mathbb{R}^{m \times m_r} \\ S \in \mathbb{R}^{n \times m_r}}}   \|M^{1/2} \cdot (AS + BR-P\hat{B})\|   \label{eq: linear bounded var2}  \\
     & \bar{r}_{3} =\|M^{1/2} S\|;  \label{eq: linear bounded var3}
    \end{align}
    \item $\hat{u}(t)$ is continuous for all $t \in [t_{0}, +\infty)$; 
    \item for every $\hat{x}_0 \in \hat{X}_{0}$, there exists some $x_0 \in X_{0}$ such that $(x_0,\hat{x}_0) \in \text{Proj}_{X\times \hat{X}}\mathcal{V}_{g}$. 
\end{enumerate}
Then $V_g(x, \hat{x}, \hat{u})$ serves as a simulation function, and $\mathcal{V}_{g}$ constitutes an $\varepsilon$-gAAS relation when the interface function $u_g(x, \hat{x}, \hat{u})$ is applied.
\end{theorem}

\begin{proof}
We demonstrate that $\mathcal{V}_{g}$ is an $\varepsilon$-gAAS relation by verifying conditions \textit{(i)}--\textit{(iii)} in Definition \ref{def:cs-gAAS}. 
Because the matrix $M$ in \eqref{eq: linear simulation function} is chosen such that $C^{\T}C \preceq M$, and $P$ and $S$ satisfy the constraints $CP = \hat{C}$ and $CS = 0_{p \times m_{r}}$ according to  Assumption \ref{def: linear conditions}, it follows that, for any $(x,\hat{x},\hat{u}) \in \mathcal{V}_{g}$, one can verify
\begin{align}
    \mathbf{d}(h(x), \hat{h}(\hat{x})) &= \| h(x) - \hat{h}(\hat{x})\| \nonumber \\
    &=\| Cx - \hat{C} \hat{x}\| = \| Cx - CP \hat{x} - CS \hat{u}\| \nonumber \\
    &=\sqrt{(x-P\hat{x}-S\hat{u})^{\T}C^{\T}C(x-P\hat{x}-S\hat{u})} \nonumber \\
    &\le \sqrt{(x-P\hat{x}-S\hat{u})^{\T}M(x-P\hat{x}-S\hat{u})} \le \varepsilon.
\end{align}
Moreover, because for every $\hat{x}_{0} \in \hat{X}_{0}$ there exist $x_{0} \in X_{0}$ and $\hat{u}_{0} \in \hat{U}$ such that $(x_{0}, \hat{x}_{0}, \hat{u}_{0}) \in \mathcal{V}_{g}$, conditions \textit{(i)} and \textit{(ii)} in Definition \ref{def:cs-gAAS} are fulfilled.  
We now move on to verifying condition \textit{(iii)}.

First, suppose that $\hat{u}$ is continuous.
Let $e(t) = x(t) - P\hat{x}(t) - S\hat{u}(t)$. Then we obtain
\begin{align} \label{eq: linear dot e}
    \dot{e}(t) = &(A + BK)e(t) + (AP - P\hat{A} + BQ)\hat{x}(t)\nonumber \\ 
    &+ (AS + BR - P\hat{B})\hat{u}(t) - S\dot{\hat{u}}(t).
\end{align}
Then, $\dot{V}_{g}(x,\hat{x},\hat{u})$ can be expressed as
\begin{align}\label{eq: linear dot Ve}
    \dot{V}_{g}(x,\hat{x},\hat{u}) 
    & = \frac{1}{2} \big(e^{\T}(t)Me(t)\big)^{-\frac{1}{2}} \big[\dot{e}^{\T}(t)Me(t) + e^{\T}(t)M\dot{e}(t)\big] \nonumber \\
    & = \frac{\dot{e}^{\T}(t)Me(t) + e^{\T}(t)M\dot{e}(t)}{2\sqrt{e^{\T}(t)Me(t)}}. \nonumber 
\end{align}
Moreover,
\begin{align}
    &~~~~~\dot{e}^{\T}(t)Me(t)  + e^{\T}(t)M\dot{e}(t) \nonumber \\
    & = e^{\T}(t)\big[(A + BK)^{\T}M + M(A + BK)\big]e(t) \nonumber \\
    &~~~~ + 2e^{\T}(t)M\big[(AP - P\hat{A} + BQ)\hat{x}(t) \nonumber \\
    &~~~~~~~~~~~~~~~~~ + (AS + BR - P\hat{B})\hat{u}(t) - S\dot{\hat{u}}(t)\big].
\end{align}
Since $M$ and $K$ fulfill the Lyapunov inequality given in \eqref{eq: linear Lyapunov ineq}, and $Q$, $R$, and $S$ are selected according to \eqref{eq: linear QRS select 1}--\eqref{eq: linear QRS select 2}, the expression in \eqref{eq: linear dot Ve} can be bounded by
   {
\allowdisplaybreaks
\begin{align}
 &\dot{V}_{g}(x,\hat{x},\hat{u})  \le \frac{\begin{matrix}
-a_{1}e^{\T}Me + 2e^{\T}M[(AP -  P\hat{A} + BQ)\hat{x}(t)  \\
~~~~~+(AS + BR-P\hat{B})\hat{u}(t) -S\dot{\hat{u}}]
\end{matrix}}{2\sqrt{e^{\T}Me}} \nonumber \\
&\le -\frac{a_{1}}{2}\sqrt{e^{\T}Me } +  \frac{2e^{\T} M^{1/2}}{2\sqrt{e^{\T}Me}} \cdot \{M^{1/2} \cdot \nonumber \\
&~~[(AP -  P\hat{A} + BQ)\hat{x}(t)  +(AS + BR-P\hat{B})\hat{u}(t) -S\dot{\hat{u}}(t)] \} \nonumber \\
& \le -\frac{a_{1}}{2}\sqrt{e^{\T}Me } + \bar{r}_{1}\|\hat{x}(t)\| + \bar{r}_{2} \|\hat{u}(t)\| + \bar{r}_{3}\| \dot{\hat{u}}(t)\|,
\end{align}}where $\bar{r}_{1}$, $\bar{r}_{2}$, and $\bar{r}_{3}$ are given by \eqref{eq: linear bounded var1}--\eqref{eq: linear bounded var3}. Because $\bar{r}(\hat{x}, \hat{u}, \dot{\hat{u}}) \le \bar{r}_{\max}$ holds, the comparison lemma \citep{khalil2002nonlinear} implies that
\begin{equation}
    V_{g}(x,\hat{x},\hat{u}) \le e^{-\frac{a_{1}}{2}t} V_{g}(x_{0},\hat{x}_{0},\hat{u}_{0}) + (1-e^{-\frac{a_{1}}{2}t})\frac{2\bar{r}_{max}}{a_{1}}.
\end{equation}
Since $\frac{2\bar{r}_{\max}}{a_1} \le \varepsilon$, it follows that
\begin{equation}
     V_{g}(x(t),\hat{x}(t),\hat{u}(t)) \le \varepsilon,
\end{equation}
for all initial conditions $(x_{0},\hat{x}_{0},\hat{u}_{0}) \in \mathcal{V}_{g}$ and for all $t \ge t_{0}$.
Hence, condition~\textit{(iii)} in Definition~\ref{def:cs-gAAS} holds, implying that $V_{g}(x,\hat{x},\hat{u})$ qualifies as a simulation function according to Definition~\ref{def: simulation function}.
We now verify that $u(t) \in U$ by applying the interface function \eqref{eq: linear explicit interface function} with $\hat{u}(t) \in \hat{U}$.
Using $\lambda_{\min}(M)\|e\|^2 \le e^{\T}Me \le \varepsilon^2$, we obtain
\begin{align}
       \|u(t)\| &\le \|K\|\|e\| +\|Q \hat{x}\| + \|R\hat{u}\| \nonumber \\ &\le \|K\| \frac{\varepsilon}{\sqrt{\lambda_{min}(M)}} +\|Q \hat{x}\| + \|R\hat{u}\| \le b.
\end{align}
Therefore, the input constraint $u(t) \in U$ is satisfied based on \eqref{eq: input constraint}.
Thus, Theorem~\ref{theorem: linear valid simulation and interface} is proved.
\end{proof}

Next, the following corollary addresses the case in which the input $\hat{u}$ may be discontinuous. 
\begin{corollary} \label{corollary}
Let the jump value of $\hat{u}$ at time $\tau$ be denoted by $\hat{\delta}(\tau)$, as in Definition~\ref{def: jump value and set}.
Consider the same setting as in Theorem~\ref{theorem: linear valid simulation and interface}, except that condition (ii) is relaxed to permit $\hat{u}(t)$ to be discontinuous, in accordance with Definition~\ref{def: jump value and set}. Assume that, for any jump time $\tau > t_{0}$, the jump value $\hat{\delta}(\tau)$ satisfies
\begin{equation} \label{eq: jump set restrict}
\hat{\delta}(\tau)^{\T} S^{\T} M S \hat{\delta}(\tau) \le \left( \varepsilon - {\omega(\tau)} \right)^2,
\end{equation}
where
\begin{equation}
\omega(\tau) := e^{- \frac{a_1}{2} \tau} V_g(x_0, \hat{x}_0, \hat{u}_0) + \left( 1 - e^{- \frac{a_1}{2} \tau} \right) \frac{2\bar{r}_{max}}{a_1}.
\end{equation}
Then $V_g(x, \hat{x}, \hat{u})$ serves as a simulation function, and $\mathcal{V}_{g}$ defines an $\varepsilon$-gAAS relation when the interface function $u_g(x, \hat{x}, \hat{u})$ is employed.
\end{corollary}
\begin{proof}
The proof proceeds analogously to that of Theorem~\ref{theorem: linear valid simulation and interface} when checking conditions~\textit{(i)} and \textit{(ii)} in Definition~\ref{def:cs-gAAS}. 
Hence, it remains only to handle the discontinuities of $\hat{u}$ in relation to condition~\textit{(iii)}. 
Because the jump values $\hat{\delta}(\tau)$ are restricted by \eqref{eq: jump set restrict}, 
for any jump time $\tau>t_0$, since $x$ and $\hat{x}$ remain
continuous and
$\hat{\delta}(\tau)=\hat{u}(\tau^{+})-\hat{u}(\tau^{-})$, we have
\begin{equation}
    e(\tau^{+})
    =
    e(\tau^{-})-S\hat{\delta}(\tau).
\end{equation}
Therefore, by the triangle inequality,
\begin{align}
&V_g\bigl(x(\tau^{+}),\hat{x}(\tau^{+}),\hat{u}(\tau^{+})\bigr)
\nonumber\\
&=
\sqrt{
\bigl(e(\tau^{-})-S\hat{\delta}(\tau)\bigr)^{\T}
M
\bigl(e(\tau^{-})-S\hat{\delta}(\tau)\bigr)}
\nonumber\\
&\leq
\sqrt{e^{\T}(\tau^{-})Me(\tau^{-})}
+
\sqrt{
\hat{\delta}^{\T}(\tau)S^{\T}MS\hat{\delta}(\tau)}
\nonumber\\
&=
V_g\bigl(x(\tau^{-}),\hat{x}(\tau^{-}),\hat{u}(\tau^{-})\bigr)
+
\sqrt{
\hat{\delta}^{\T}(\tau)S^{\T}MS\hat{\delta}(\tau)}
\nonumber\\
&\leq
\omega(\tau)
+
\sqrt{
\hat{\delta}^{\T}(\tau)S^{\T}MS\hat{\delta}(\tau)}
\leq\varepsilon.
\end{align}
Hence, for any jump time
$\tau>t_0$,
\begin{equation}
V_g\bigl(x(\tau^{+}),\hat{x}(\tau^{+}),
\hat{u}(\tau^{+})\bigr)\leq\varepsilon.
\end{equation}
Consequently, by replicating the argument used to verify 
condition~\textit{(iii)} in the proof of Theorem~\ref{theorem: linear valid simulation and interface}, we derive
\begin{equation}
    V_{g}(x(t), \hat{x}(t), \hat{u}(t)) \le \varepsilon
\end{equation}
for all $(x_{0}, \hat{x}_{0}, \hat{u}_{0}) \in \mathcal{V}_{g}$ and for all $t \ge t_{0}$, when the interface function \eqref{eq: linear explicit interface function} is applied and $\hat{u}$ has discontinuities constrained by \eqref{eq: jump set restrict}.
Therefore, condition~\textit{(iii)} in Definition \ref{def:cs-gAAS} holds, and this completes the proof of Corollary~\ref{corollary}.
\end{proof}

\begin{remark}
Determining matrices $S$ and $R$ that satisfy $CS = 0$ and $AS + BR - P\hat{B} = 0$, thereby enforcing $\bar{r}_{2} = 0$ in \eqref{eq: linear bounded var2}, is equivalent to solving  
\[
\begin{bmatrix}
    A & B \\
    C & 0
\end{bmatrix}
\begin{bmatrix}
    S \\ R
\end{bmatrix}
=
\begin{bmatrix}
    P\hat{B} \\ 0
\end{bmatrix}.
\]
A solution to this equation exists whenever the control system $\Sigma$ in \eqref{eq: linear-concrete} has no invariant zero at zero.  
Equivalently, there is no nonzero pair $(x, u) \neq (0,0)$ such that  
$A x + B u = 0$ and $C x = 0$,
meaning that the system output cannot remain identically zero while the state and input evolve according to the dynamics at zero frequency.  
In hierarchical control, if the concrete system meets this condition, it guarantees that every nontrivial trajectory $(x(t), u(t))$ produces an output that can be distinguished from zero.  
Observe that choosing $S = 0$ recovers the result of \citep{girard2009hierarchical}, so our result represents a strict generalization.
\end{remark}
\vspace{-3pt}
\section{Case study} \label{case study}
\vspace{-5pt}
In this section, we examine the concrete system $\Sigma$ and the abstract system $\hat{\Sigma}$, represented by \eqref{case: linear concrete model} and \eqref{case: linear abstract model}, respectively. We then apply the result of Theorem \ref{theorem: linear valid simulation and interface} to establish a $0.5$-gAAS relation, that is, $\hat{\Sigma}_{} \preceq_{g\mathcal{AS}}^{0.5} \Sigma_{}$.
Let
\begin{equation}
    M = \begin{bmatrix} 3.9544 & 1.1805 \\[2pt] 1.1805 & 4.2262 \end{bmatrix}, ~
K = [-1.3298, -1.4108], ~
a_{1} = 0.5,
\end{equation}
be chosen so that \eqref{eq: linear Lyapunov ineq} is satisfied.  
Then, applying Assumption~\ref{def: linear conditions} together with Theorem~\ref{theorem: linear valid simulation and interface}, we obtain
$Q = 0$, $R = 0$, $S =[0;1]$, $P=[1;0]$, and $\|\dot{\hat{u}}\| \le 0.0486$,
with $\bar{r}_{1} = \bar{r}_{2} = 0$ and $\bar{r}_{3} = 2.0558$ in~\eqref{eq: linear bounded var1}--\eqref{eq: linear bounded var3}.

Consider the initial condition
$(x(0), \hat{x}(0), \hat{u}(0)) = (40, 40.1,$ $ -0.0401) \in 
\{ (x, \hat{x}, \hat{u}) \in X \times \hat{X} \times \hat{U} \mid V_{g}(x, \hat{x}, \hat{u}) \le 0.5 \}$,
where $V_{g}$ is given in~\eqref{eq: linear simulation function}.  
We consider an abstract controller of the form $\hat{u} = -\hat{k} \hat{x}$, deliberately constructed to be discontinuous, in order to illustrate how the proposed control refinement method performs when the abstract system is subject to discontinuous inputs. We set
\[
\hat{k} =
\begin{cases}
0.001  & \hat{x} \in [30, \hat{x}(0)],
~~~~0.0013 ~~~ \hat{x} \in [20, 30), \\[2pt]
0.002  & \hat{x} \in [10, 20), 
~~~~~~~0.004 ~~~~~ \hat{x} \in [0, 10),
\end{cases}
\]
so that both $\|\dot{\hat{u}}\| \le 0.0486$ and~\eqref{eq: jump set restrict} hold.
Consequently, the trajectories of $y$, $\hat{y}$ and $\|y - \hat{y}\|$ are depicted in Fig.~\ref{fig: linear_I}, showing that $\hat{\Sigma}_{} \preceq_{g\mathcal{AS}}^{0.5} \Sigma_{}$, since the output deviation satisfies $\|y - \hat{y}\| \le 0.5$ under the discontinuous input.

Furthermore, to compare our method with that of \citep{girard2009hierarchical}, we choose the control input $\hat{u} = 0.02t$ for $t \in [0,50]$ and $\hat{u} = 1$ for $t \in (50,+\infty)$ so as to enforce a $0.5$-gAAS relation from the initial condition $(x(0), \hat{x}(0), \hat{u}(0)) = (40, 40.1, 0)$.
Under this setup, since $ \|\hat{u}\| \le 1$, \eqref{eq: input constraint} implies that $\|u \| \le 0.5690$.
It can then be checked that no $\hat{\Sigma}_{} \preceq_{\mathcal{AS}}^{0.5} \Sigma_{}$ exists.
The corresponding outputs $y$, $\hat{y}$, and the error $\|y - \hat{y}\|$ for both approaches are plotted in Fig.~\ref{fig: linear_I_compare} and Fig.~\ref{fig: linear_I_papas}.
As illustrated there, the method of \citep{girard2009hierarchical} does not guarantee the $0.5$-gAAS relation for this controller $\hat{u}$, because it fails to ensure that the output error satisfies $\|y - \hat{y}\| \le 0.5$ at all times.
\begin{figure}[htbp]
\vspace{0pt}
    \centering
    \subfigure[The $y$ and $\hat{y}$]
    {
    \includegraphics[width=1\hsize]{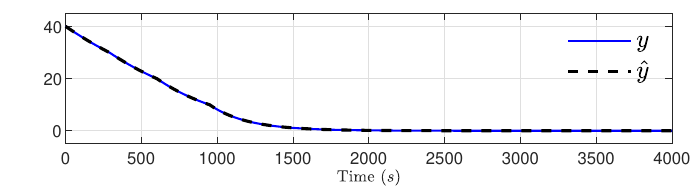}
    \label{fig: linear_I_y_yhat}
    }
    \subfigure[$\|y - \hat{y} \|$]
    {
    \includegraphics[width=1\hsize]{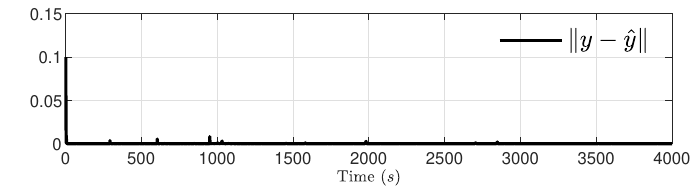}
    \label{fig: linear_I_d}
    }
    \caption{The simulation result of the proposed approach. }
    \label{fig: linear_I} 
\end{figure}

\begin{figure}[htbp]
\vspace{-2pt}
    \centering
    \subfigure[The $y$ and $\hat{y}$]
    {
    \includegraphics[width=1\hsize]{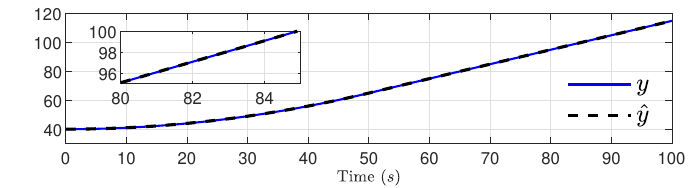}
    \label{fig: linear_I_y_yhat_compare}
    }\vspace{-1pt}
    \subfigure[$\|y - \hat{y} \|$]
    {
     \includegraphics[width=1\hsize]{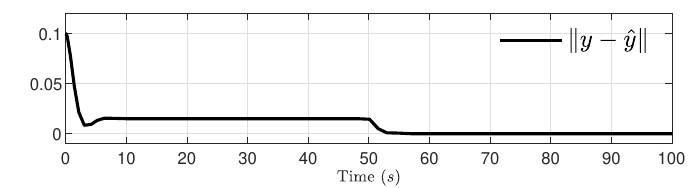}
    \label{fig: linear_I_d_compare}
    }\vspace{-1pt}
    \caption{The simulation result of the proposed approach. }
    \vspace{3pt}
    \label{fig: linear_I_compare} 
\end{figure}

\begin{figure}[htbp]
\vspace{-2pt}
    \centering
    \subfigure[The $y$ and $\hat{y}$]
    {
    \includegraphics[width=1\hsize]{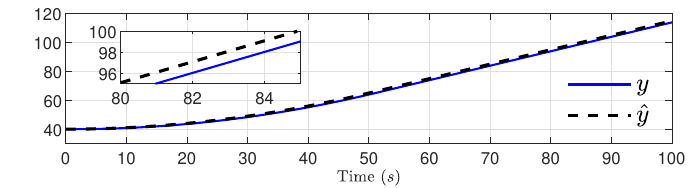}
    \label{fig: linear_I_y_yhat_papas}
    }
    \subfigure[$\|y - \hat{y} \|$]
    {
    \includegraphics[width=1\hsize]{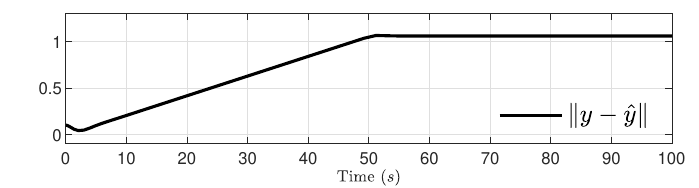}
    \label{fig: linear_I_d_papas}
    }
    \caption{The simulation result by implementing control refinement approach in \citep{girard2009hierarchical}. }
    \vspace{3pt}
    \label{fig: linear_I_papas} 
\end{figure}

\section{Conclusion} \label{conclusion}
\vspace{-3pt}
This paper presents the general approximate alternating simulation relation (\emph{$\varepsilon$-gAAS relation}) for continuous-time systems and investigates its characteristics. Building on this relation, a hierarchical control strategy is proposed for continuous-time linear systems. The performance of the approach is illustrated through case studies and benchmarked against existing methods. Future research will focus on extending the control refinement framework to continuous-time nonlinear systems within the $\varepsilon$-gAAS setting.


\bibliography{ifacconf}             
\end{document}